\pdfoutput=1
\RequirePackage{ifpdf}
\ifpdf 
\documentclass[pdftex]{sigma}
\else
\documentclass{sigma}
\fi

\numberwithin{equation}{section}

\newtheorem{Theorem}{Theorem}[section]
\newtheorem*{Theorem*}{Theorem}

\newtheorem{Lemma}[Theorem]{Lemma}
\newtheorem{Proposition}[Theorem]{Proposition}
 { \theoremstyle{definition}

\newtheorem{Example}[Theorem]{Example}
\newtheorem{Remark}[Theorem]{Remark} }

\usepackage{tikz}
\usepackage{arydshln}
\usepackage{ascmac}
\usepackage{enumerate}

\begin{document}
\allowdisplaybreaks

\newcommand{\arXivNumber}{2211.16706}

\renewcommand{\PaperNumber}{094}

\FirstPageHeading

\ShortArticleName{A $3 \times 3$ Lax Form for the $q$-Painlev\'e Equation of Type $E_6$}

\ArticleName{A $\boldsymbol{3 \times 3}$ Lax Form for the $\boldsymbol{q}$-Painlev\'e Equation\\ of Type $\boldsymbol{E_6}$}

\Author{Kanam PARK}

\AuthorNameForHeading{K.~Park}

\Address{National Institute of Technology, Toba College, 1-1, Ikegami-cho, Toba-shi, Mie, Japan}
\Email{\href{mailto:paku-k@toba-cmt.ac.jp}{paku-k@toba-cmt.ac.jp}}

\ArticleDates{Received December 01, 2022, in final form November 05, 2023; Published online November 18, 2023}

\Abstract{For the $q$-Painlev\'e equation with affine Weyl group symmetry of type $E_6^{(1)}$, a~$2\times 2$ matrix Lax form and a second order scalar lax form were known. We give a new~$3\times 3$ matrix Lax form and a third order scalar equation related to it. Continuous limit is also discussed.}


\Classification{14H70; 34M56; 39A13}

\section{Introduction}
The $q$-Painlev\'e equation with affine Weyl group symmetry of type $E_6^{(1)}$ was first discovered in~\cite{rgh91}.
 The well-known form of it is as follows
\begin{gather*}
T\colon\ (a_1,a_2,a_3,a_4,a_5,a_6,a_7,a_8;f,g)\to
 \big(a_1/q,a_2/q,a_3/q,a_4/q,a_5,a_6,a_7,a_8;\overline{f},\overline{g}\big), \\
\frac{\big(\overline{f}g-1\big)(fg-1)}{f \overline{f}}=\frac{(g-1/a_5)(g-1/a_6)( g-1/a_7)(g-1/a_8)}{(g-a_3)(g-a_4)}, \\
\frac{\big(\overline{f}g-1\big)\big(\overline{f}\overline{g}-1\big)}{g\overline{g} }=\frac{\big(\overline{f}-a_5\big)\big(\overline{f}-a_6\big)\big(\overline{f}-a_7\big)\big(\overline{f}-a_8\big)}{\big(\overline{f}-a_1/q\big)\big(\overline{f}-a_2/q\big)},\\
q=\frac{a_1a_2}{a_3a_4a_5a_6a_7a_8},
\end{gather*}
where $f$ and $g$ are dependent variables, $a_1, a_2, \dots, a_8$ are parameters, $q$ is a constant and the overline symbol ``$\ \bar{}\ $'' denotes the discrete time evolution.

In previous works, the following Lax forms for the $q$-Painlev\'e equation for type $E_6^{(1)}$ has been obtained.
In \cite{sakai06}, a $2 \times 2$ matrix Lax pair was first derived as a reduction of the $q$-Garnier system.
The other approach gives a $2 \times 2$ matrix Lax pair \cite{dk20, kni16}.
In \cite{ymd11}, a second order scalar Lax form was obtained as a reduction from the $q$-Painlev\'e equation of type $E_8$.
The relation between these Lax forms was given in \cite{wo12}.

In this article, we give a new Lax form with $3 \times 3$ matrix Lax pair.
We derive such a Lax form as a special case of the system investigated in our previous work \cite{park20}.
As a result, we derive an equation which is equivalent to the $q$-Painlev\'e equation of type $E_6^{(1)}$ \cite{rg00,rgtt01}.

In the previous work \cite{park20}, we defined a nonlinear $q$-difference system as a connection preserving deformation of the following linear equation
\begin{gather}
\Psi (qz)=\Psi(z)A(z), \qquad A(z)=DX_1^{\varepsilon_1}(z)X_2^{\varepsilon_2}(z)\cdots X_M^{\varepsilon_M}(z),\nonumber\\
D=\operatorname{diag}[d_1,d_2,\dots,d_N],\nonumber\\
X_i(z)=\operatorname{diag} [u_{1,i},u_{2,i},\dots, u_{N,i}]+\Lambda, \qquad
\Lambda=
\begin{bmatrix}
 0&1& &{\bf O}\\
 & \ddots&\ddots&\\
 & & &1\\
z& &&0
\end{bmatrix},\label{laxmn}
\end{gather}
where the exponents are $\varepsilon_i=\pm 1$ $(1 \leq i \leq M)$,
 $ u_{j,i}$ ($1 \leq j \leq N$) are dependent variables and~$c_i$,~$d_j $ are parameters which satisfy
\[
\prod_{j=1}^N u_{j,i}=c_i.
\]
Since one can exchange the order of matrices $X_i^{\pm 1}$ by suitable rational transformations of variables $u_{j,i}$, the equation \eqref{laxmn}
essentially depends on $M_+$, $M_-$, where $M_{\pm}=\# \{ \varepsilon_i | \varepsilon_i=\pm 1 \}$.

The contents of this paper is as follows.
In Section \ref{l33}, we set up a linear $q$-difference equation~\eqref{lax}, which is a case of $(M_+,M_-,N)=(3,0,3)$ for the equation~\eqref{laxmn}.
And we discuss about its two deformations.
One deformation gives rise to a
well known form of the $q$-Painlev\'e equation of type \smash{$E_6^{(1)}$},
and the other deformation gives an equation for a non-standard direction.
Namely, it does not give a $q$-shift deformation for parameters.
In Section \ref{scle6}, we derive a scalar equation from the $3 \times 3$ matrix equation \eqref{lax} and consider its characteristic properties.
In Section \ref{blc}, we study continuous limit of our constructions and its relation to the Boalch's Lax pair \cite{boalch09}.
In Appendix \ref{rvt12}, we give deformations considered in Section \ref{l33} on root variables.
In Appendix \ref{b}, we give explicit forms of coefficients of a single linear $q$-difference equation derived in Section \ref{scle6}.
\begin{Remark}
The equation \eqref{laxmn} in case $(M_+,M_-,N)=(2n+2,0,2)$ is known that it is equivalent to the linear $q$-difference equation related to the $2n$-dimensional $q$-Garnier system~\cite{sakai05}.
And in case of $(M_+,M_-,N)=(2,0,2n+2)$, the equation \eqref{laxmn}
is also equivalent to the linear $q$-difference equation related to the $2n$-dimensional system $q$-$P_{(n+1,n+1)}$ \cite{suzuki15}.
In both cases, when $n=1$, they give rise to the $q$-$P_{\rm VI}$ equation \cite{js96}.
We note that the equation \eqref{lax} is coincide to the $q$-difference linear equation of the Lax form in \cite{szk21} in case of ($m$, $n$)=($3$, $1$).
 \end{Remark}

\section[A 3 times 3 matrix Lax for]{A $\boldsymbol{3 \times 3}$ matrix Lax form}\label{l33}

In this section, we consider two types of deformations for
the linear $q$-difference equation \eqref{laxmn} in a case $(M_+,M_-,N)=(3,0,3)$.

We consider the connection preserving deformation for the following $q$-difference equation for an unknown function $\Psi(z)=[\Psi_1(z),\Psi_2(z),\Psi_3(z)]$:
\begin{gather}
\Psi (qz)=\Psi(z)A(z)=\Psi(z)A(z,t), \nonumber\\
A(z)=DX_1(z)X_2(z)X_3(z)
=\begin{bmatrix}
b_1d_1&*&*\\
0&b_2d_2&*\\
0&0&b_3d_3
\end{bmatrix}
+\begin{bmatrix}
d_1&0&0\\
*&d_2&0\\
*&*&d_3
\end{bmatrix}z,\label{lax}
\end{gather}
where the matrices $D$ and $X_i(z)$ $(1\leq i \leq 3)$ stand for
\begin{gather*}
D=\operatorname{diag}[d_1,d_2,d_3],\qquad
X_i(z)=\operatorname{diag} [u_{1,i},u_{2,i},u_{3,i}]+\Lambda, \qquad
\Lambda=
\begin{bmatrix}
 0&1&0 \\
 0& 0& 1\\
z& 0&0
\end{bmatrix},
\end{gather*}
and $ u_{j,i} $ ($1 \leq i$, $j \leq 3$) are dependent variables and $c_i$, $d_j $ are parameters which satisfy
\begin{gather}\label{para}
\displaystyle \prod_{i=1}^3 u_{j,i}=b_j,\quad 1\leq j \leq 3,\qquad
\displaystyle \prod_{j=1}^3 u_{j,i}=c_i,\quad 1\leq i \leq 3,\qquad
b_1b_2b_3=c_1c_2c_3.
\end{gather}
The first equation in \eqref{para} is equivalent to that the characteristic exponents at $z=0$ of the equation \eqref{lax} are $b_j d_j$.
The second equation in \eqref{para} is equivalent to the following condition:
\begin{equation}\label{det}
|A(z)|=d_1d_2d_3 (z+c_1)(z+c_2)(z+c_3).
\end{equation}
Through a gauge transformation by a $3\times 3$ diagonal matrix
we can take two components in \eqref{lax} as 1.
In the following, we use this kind of gauge fixings in case by case.
By the condition \eqref{det}, two of the remaining four components are determined by other components and parameters~$b_j$,~$c_i$,~$d_j$.

In this article, we will consider two deformations $T_1$ and $T_2$ for the equation \eqref{lax} which act
 on parameters $b_j$, $c_i$, $d_j$ as
 \begin{gather*}
 T_1\colon\ ({b_1},{b_2},{b_3},{c_1},{c_2},{c_3},
 {d_1},{d_2},{d_3}) \to \left(b_1,qb_2,\frac{b_3}{q},c_1,c_2,c_3,qd_1,d_2,qd_3\right), \\
 T_2\colon\ ({b_1},{b_2},{b_3},{c_1},{c_2},{c_3},
 {d_1},{d_2},{d_3}) \to \left(\frac{b_1 d_1}{d_2},\frac{b_2 d_2}{d_3},b_3,c_1,c_2,c_3 q,d_2,d_3,\frac{d_1}{q}\right).
 \end{gather*}
 \subsection[The deformation T\_1]{The deformation $\boldsymbol{T_1}$}\label{pd}
 We will show that the deformation $T_1$ gives rise to the standard form of the $q$-$E_6^{(1)}$.
 We consider the following $q$-difference linear equation:
 \begin{gather}\label{e6a}
 \Psi(qz)=\Psi(z) A(z), \qquad A(z) =\begin{bmatrix}
 b_1d_1 & 1 & v_1 \\
 0 & b_2d_2 & 1 \\
 0 & 0 & b_3d_3
\end{bmatrix}
+\begin{bmatrix}
 d_1 & 0 & 0 \\
 v_2 & d_2 & 0 \\
 v_3 & v_4 & d_3
 \end{bmatrix}z, \\
 |A(z)|=d_1d_2d_3(z+c_1)(z+c_2)(z+c_3).\label{dett}
 \end{gather}
Although there are many ways of gauge fixings,
fixing as the equation \eqref{e6a} makes relatively easier to find a new change of variables from indeterminate points of time evolution equations $\overline{v_i}$ ($\overline{*}=T_1(*)$).

 As a connection preserving deformation for the equation \eqref{e6a}, \eqref{dett} we take the following deformation for parameters
 \begin{equation}\label{deqq}
 T_1\colon\ (b_1,{b_2},{b_3},c_1,c_2,c_3,
 {d_1},d_2,{d_3})
 \to \left( b_1,qb_2,\frac{b_3}{q},c_1,c_2,c_3,qd_1,d_2,qd_3\right).
 \end{equation}
Then, there is a matrix $B(z)$ which satisfies the following deformation equation:
\[
T_1 \Psi(z)=\Psi(z)B(z).
\]
We derive and show the matrix $B(z)$.
First, the matrix $B(z)$ is a rational function for $z$ \cite{js96}.
In fact,
by the argument \cite{js96} to determine a coefficient matrix of a deformation equation,
the matrix~$B(z)$ is of degree one through the following:
\begin{itemize}\itemsep=0pt
\item[(i)] The parameters $c_i$ which satisfy $|A(c_i)|=0$ are constant by the deformation $T_1$ \eqref{deqq}.
Therefore, the matrix $B(z)$ does not have poles at $z=-q^k c_i$ ($k \in \mathbb{Z}$).
Namely, the matrix~$B(z)$ has poles only at $z=0$ or $z=\infty$.
\item [(ii)] The deformation $T_1$ \eqref{deqq} shifts characteristic exponents of the equation \eqref{e6a} at $z=0$ as follows:
\[
T_1\colon\ (b_1d_1, b_2d_2, b_3d_3) \to (qb_1d_1, qb_2d_2, b_3d_3).
\]
Therefore, the matrix $B(z)$ behaves nearby $z=0$ as
\[
B(z)=B_0+B_1z+\mathcal{O}\big(z^2\big).
\]
\item [(iii)] The deformation $T_1$ \eqref{deqq} shifts characteristic exponents of the equation \eqref{e6a} at $z=\infty$ as follows:
\[
T_1\colon\ (d_1, d_2, d_3) \to (qd_1, d_2, qd_3).
\]
Therefore, the matrix $B(z)$ behaves nearby $z=\infty$ as
\[
B(z)=B_1z+B_0+\mathcal{O}\bigg(\frac{1}{z}\bigg).
\]
 \end{itemize}
From the above (i)--(iii), the matrix $B(z)$ is a polynomial in $z$ of degree one.
At last, we derive an explicit form of the matrix $B(z)$.
We express the matrix $B(z)$ as follows:
\begin{equation}\label{bz}
B(z)=B_0+B_1z,\qquad B_i=\big[\beta_{j,k}^i\big]_{1\leq j, k \leq 3},\quad i=1, 2.
\end{equation}
Comparing coefficients of $z$ for a
compatibility condition equation
\[
B(z)\overline{A(z)}=A(z)B(qz), \qquad \overline{*}=T_1(*),
\]
 for the matrices $A(z)$ \eqref{e6a} and $B(z)$ \eqref{bz}, we have the following three equations:
\begin{equation}\label{abkeisu}
B_0 \overline{A_0}=A_0 B_0,\qquad B_1\overline{A_0}+B_0\overline{A_1}=qA_0 B_1+A_1B_0, \qquad B_1\overline{A_1}=qA_1B_1,
\end{equation}
where matrices $A_i$ denote coefficients of $z^i$ for the matrix $A(z)$ \eqref{e6a}.
Solving the first and the third equations of \eqref{abkeisu}, forms of the matrices $B_0$ and $B_1$ are as follows:
 	\begin{gather*}
B_0=
	\begin{bmatrix}
	 0 & 0 & \frac{\beta_{2,3}^{0}\left(b_2 d_2 v_1-b_3 d_3
 v_1-1\right)}{b_1 d_1-b_3
 d_3} \\
 0 & 0 &\beta_{2,3}^{0} \\
 0 & 0 &
 -\beta_{2,3}^{0}\left(b_2 d_2-b_3
 d_3\right)
	 \end{bmatrix}, \\
B_1=\begin{bmatrix}
\frac{\left(d_1-d_2\right)
 \beta_{2,1}^{1} }{v_2} & 0 & 0
 \\
 \beta_{2,1}^1 & 0 & 0 \\
 \frac{d_1 q v_3 \overline{v_4}
 \beta_{2,1}^1 +d_2 \left(v_2
 \overline{v_3} \beta_{3,2} ^1-q
 v_3 \overline{v_4}
 \beta_{2,1}^1 \right)+v_2
 \left(q v_4 \overline{v_4}
 \beta_{2,1}^1 -\beta_{3,2}^1 \left(d_3 q
 \overline{v_3} +\overline{v_2}
 \overline{v_4} \right)\right)}{\left(d_1-d_3\right) q v_2
 \overline{v_4} } &
 \beta_{3,2}^1 &
 -\frac{\beta_{3,2} ^1
 \left(d_2-d_3
 q\right)}{\overline{v_4} }
 \end{bmatrix}.
\end{gather*}
From the second equation of \eqref{abkeisu},
we obtain explicit forms of time evolutions $\overline{v_i}$ ($1\leq i \leq 4$), the components $\beta_{2,3}^0$ and $\beta_{2,1}^1$ of the matrix $B(z)$ \eqref{bz}.
Explicit forms of the remain components~$\beta_{2,3}^0$ and $\beta_{2,1}^1$ are as follows:
\begin{gather*}
\beta_{2,3}^0= \beta_{3,2}^1 q
 (b_1 d_1-b_3 d_3)
 (d_1-d_2-v_1
 v_2)\bigl(-b_2^2 d_2^2 q v_1 v_2 (d_3 q-d_1+v_1v_2)-b_3 d_3 d_2 q v_2 \\
\hphantom{\beta_{2,3}^0=}{}
-2 b_3^2 d_3^2 d_2 q v_1
 v_2+b_3 d_3 d_2 q v_1 v_3+b_2 d_2
 \big(b_3 d_3
 (d_1 q (-d_3 (q-1)-2 v_1v_2)\\
\hphantom{\beta_{2,3}^0=}{}
 +q v_1 v_2 (d_3 q+2 v_1v_2)+d_2 (d_3(q-1) q+(2 q-1) v_1v_2))
 +q \big(d_3
 q v_2+d_1 (v_1
 v_3-v_2)\\
\hphantom{\beta_{2,3}^0=}{}
+v_1 \bigl(-d_2
 v_3+v_2^2-v_1 v_3
 v_2\bigr)\big)\big)+b_3
 d_3 d_2 q v_4-b_3^2 d_3^2 q
 v_1^2 v_2^2-b_3 d_3 q v_1
 v_2^2-b_3 d_3^2 q v_2\\
\hphantom{\beta_{2,3}^0=}{}
 +b_3 d_1 d_3 q v_2
 +b_3^2 d_1 d_3^2 qv_1 v_2-b_3 d_1 d_3 q v_1v_3
 +b_3 d_3 q v_1^2 v_2 v_3-b_3 d_1 d_3 q v_4\\
\hphantom{\beta_{2,3}^0=}{}
 +b_3 d_3q v_1 v_2 v_4+b_1 d_1
 \big(b_3 d_3 \big(d_2 q v_1v_2-d_3 q v_1 v_2+d_2^2
 (q-1)
 -d_1 d_2 (q-1)\big)\\
\hphantom{\beta_{2,3}^0=}{}
 +b_2 d_2 q (d_2 (d_3 (-q)+d_3-v_1 v_2)+d_1 d_3 (q-1)
 +d_3 v_1 v_2)
 +q v_4 (d_1-d_2-v_1 v_2)\big)
\\
\hphantom{\beta_{2,3}^0=}{}
 -b_3^2 d_3^2 d_2^2 q+b_3^2 d_1 d_3^2 d_2 q
 +b_3 d_3 d_2 v_2+b_3^2 d_3^2 d_2 v_1 v_2
 +b_3^2 d_3^2 d_2^2\\
\hphantom{\beta_{2,3}^0=}{}
 -b_3^2 d_1 d_3^2 d_2
 +d_2 q v_3-d_1 q v_3+q v_1 v_2 v_3\bigr)^{-1},\\
 \beta_{2,1}^1= \beta_{3,2}^1\frac{q v_2
 (b_1 d_1 v_1-b_2 d_2
 v_1+1)}{v_2 (1-b_2
 d_2 v_1)+b_3 d_3
 (-d_1+d_2+v_1
 v_2)+b_1 d_1
 (d_1-d_2)}.
\end{gather*}
We remark that time evolutions $\overline{v_i}$ are independent for the component $\beta_{3,2}^1$.
Therefore, we set~$\beta_{3,2}^1=1$.

From the above, a deformation equation $T_1\Psi(z)=\Psi(z)B(z)$ is expressed as follows:
 \begin{gather}
 T_1\Psi(z)=\Psi(z)B(z),\qquad
 B(z)=
 \begin{bmatrix}
 0&0&w_1\\
 0&0&w_2\\
 0&0&w_3
 \end{bmatrix}
 +\begin{bmatrix}
 w_4&0&0\\
 w_5&0&0\\
 w_6&1&w_7
 \end{bmatrix}z, \label{t1bz}\\
 |B(z)|=(w_1w_5-w_2w_4)z^2,\nonumber\\
 T_1\colon\ ({b_1},{b_2},{b_3},{c_1},{c_2},{c_3},
 {d_1},{d_2},{d_3};v_1,v_4)\nonumber \\
 \phantom{ T_1\colon}{}\qquad\to \left(b_1,qb_2,\frac{b_3}{q},c_1,c_2,c_3,qd_1,d_2,qd_3;\overline{v_1},\overline{v_4}\right).\label{deq1}
 \end{gather}
 \begin{Theorem}
 Through a compatibility condition of the equations \eqref{e6a}, \eqref{t1bz}, \eqref{deq1},
 \begin{equation}\label{ryrt1}
 B(z) \overline{A(z)}=A(z)B(qz),
 \end{equation}
 we obtain the following equations:
\begin{gather}
\frac{(fg-1)\big(\overline{f}g-1\big)}{f \overline{f} }=\frac{(g-1/c_1)(g-1/c_2)(g-1/c_3)(g-d_2/b_3d_3)}{(g-1/b_2)(g-d_2/b_1d_1)},\nonumber \\
\frac{\big(\overline{f}g-1\big)\big(\overline{f}\overline{g}-1\big)}{g\overline{g}}
=\frac{\big(\overline{f}-c_1\big)\big(\overline{f}-c_2\big)\big(\overline{f}-c_3\big)\big(\overline{f}-b_3d_3/d_2\big)}{\big(\overline{f}-b_3d_3/d_1\big)\big(\overline{f}-b_3/q\big)},\label{e6pade}
\end{gather}
 where
 \begin{gather}\label{e6fg}
 f=\frac{b_3 d_3}{d_3-v_1 v_4},\qquad
 g= \frac{d_2 v_1}{b_2 d_2
 v_1-1},
 \end{gather}
 and $\overline{*}$ stands for $T_1(*)$.
 \end{Theorem}
 The equation \eqref{e6pade} is the well known form of the $q$-Painlev\'e equation of type $E_6$ \cite{rg00,rgtt01} (see also \cite{kny17}).
 \begin{proof}The result is obtained by a direct computation of the compatibility condition of \eqref{ryrt1}.
Since the computation is rather heavy, we will give a comment how to do it efficiently.
Though the $2$ variables~$v_2$,~$v_3$ can be represented by the rational functions of the remaining two variables $v_1$, $v_4$ by the relation \eqref{dett},
 it is more efficient to do this elimination after the calculation of the compatibility condition \eqref{ryrt1} in $4$ variables, and then reduce it to 2 variables.
In this way,
we get the following time evolutions for $v_1$, $v_4$ as rational functions of $v_1$, $v_4$:
 \begin{equation}\label{v14}
 \overline{v_1}=\frac{C_1(v_1,v_4)C_2(v_1,v_4)}{D_1(v_1,v_4)D_2(v_1,v_4)},\qquad \overline{v_4}=\frac{C_3(v_1,v_4)}{D_3(v_1,v_4)D_4(v_1,v_4)},
 \end{equation}
 where
 \begin{gather*}
 C_3(v_1,v_4), D_2(v_1,v_4)\colon\ \text{polynomials in }v_1, \ v_4 \text{ of degree } (4, 3),\\
 C_2(v_1, v_4)\colon\ \text{a polynomial in } v_1, \ v_4 \text{ of degree } (3, 1),\\
 C_1(v_1, v_4), D_3(v_1,v_4), D_4(v_1,v_4)\colon\ \text{polynomials in }v_1, \ v_4 \text{ of degree } (2, 1),\\
D_1(v_1,v_4)\colon\ \text{a polynomial in }v_1, \ v_4 \text{ of degree } (1, 0).
 \end{gather*}
The remaining task is to rewrite the equation \eqref{v14} to \eqref{e6pade}.
A useful way to solve it is to look at the singularities of the equations~\cite{antake18}.
Namely, we investigate the points at which the right hand side of the equations \eqref{v14} are indeterminate.
 We focus on the
 equation $D_3(v_1, v_4)$ of them
 \begin{equation}\label{d3}
 D_3(v_1,v_4)= d_3 (-b_2 d_2 v_1+b_3 d_2 v_1+1 )+v_1 v_4 (b_2 d_2 v_1-1 ).
 \end{equation}
 Investigating common zero points of the equation \eqref{d3} and the other polynomials $C_k(v_1, v_4)$, $D_l(v_1, v_4)$, we find 4 indeterminate points as follows:
 \begin{gather}
 (v_1,v_4)=\left(-\frac{1}{(u^{-1}-b_2)d_2},-\frac{\big(u^{-1}-b_3\big)\big(u^{-1}-b_2\big)d_3d_2}{u^{-1}}\right),\nonumber\\
 u=c_1^{-1},c_2^{-1},c_3^{-1},d_2/b_3d_3.\label{v2fg}
 \end{gather}
 The other $4$ points are the following:
 \begin{gather*}
 (v_1,v_1v_4)=(\infty,\infty),
 \qquad \left(\frac{1}{b_2d_2-b_1d_1},\infty\right),\qquad
 \left(\frac{1}{b_2d_2},d_3-d_1\right),\qquad
 \left(\frac{1}{b_2d_2},0\right).
 \end{gather*}
In view of the form of the points in \eqref{v2fg},
 we define the variables
 $f$, $g$ as follows:
 \begin{gather}\label{v2f}
 (v_1,v_1v_4)=\left( -\frac{1}{\big(g^{-1}-b_2\big)d_2}, \frac{(f-b_3)d_3}{f}\right).
 \end{gather}
 By the transformation \eqref{v2f}, the equation $D_3(v_1,v_4)=0$ is transformed to an equation~${fg=1}$.
Through the correspondence \eqref{v2f} and the time evolution equations for the variables $v_1$,~$v_4$ \eqref{v14}, we have time evolution equations for $f$, $g$ \eqref{e6pade}.
 \end{proof}

\subsection[The deformation T\_2]{The deformation $\boldsymbol{T_2}$ }\label{33}
In this subsection, we take a deformation equation which is one of
 that considered in the previous work \cite{park20}.
It corresponded to the permutations of the matrices $X_i(z)^{\pm 1}$.
 We consider the following Lax pair:
 \begin{gather}
\Psi (qz)=\Psi(z)A(z),\nonumber \\
 A(z)=DX_1(z)X_2(z)X_3(z)=\begin{bmatrix}
 b_1d_1 & 1 & * \\
 0 & b_2d_2 & 1 \\
 0 & 0 & b_3d_3
\end{bmatrix}
+\begin{bmatrix}
 d_1 & 0 & 0 \\
 * & d_2 & 0 \\
 * & * & d_3
 \end{bmatrix}z,\label{mn33}\\
 T_2{\Psi}(z)=\Psi(z)B(z),\qquad B(z)=X_3(z/q)^{-1},\qquad
 |B(z)|=\frac{q}{z+qc_3},\nonumber \\
 T_2\colon\ ({b_1},{b_2},{b_3},{c_1},{c_2},{c_3},
 {d_1},{d_2},{d_3};x,y)\nonumber \\
 \phantom{T_2\colon}{} \qquad \to \left(\frac{b_1 d_1}{d_2},\frac{b_2 d_2}{d_3},b_3,c_1,c_2,c_3 q,d_2,d_3,\frac{d_1}{q};\overline{x},\overline{y}\right),\label{deq2}
 \end{gather}
 where we define variables $x$, $y$ and gauge freedom $w_1$, $w_2$ with the variables $u_{j,i}$ in \eqref{mn33}
 as follows:
 \begin{equation}\label{swk}
 x=\frac{u_{1,1} u_{1,2} u_{2,1}
 u_{2,2}}{u_{1,1}+u_{2,2}},\qquad y=\frac{1}{u_{1,1} u_{1,2}
 \left(u_{2,1}+u_{3,2}\right)},\qquad w_1=u_{1,1},\qquad w_2=u_{1,3}.
 \end{equation}
 \begin{Theorem} Through a compatibility condition of the equations \eqref{mn33}, \eqref{deq2}, \eqref{swk},
 \begin{equation}\label{ryrt}
 B(z) \overline{A(z)}=A(z)B(qz),
 \end{equation}
 we obtain the following equations:
 \begin{equation}\label{fup33}
\frac{ \overline{x}+c_1}{\overline{x}+c_2}=\frac{E_1E_2}{F_1F_2},\qquad
\frac{c_1\overline{y}+1}{c_2\overline{y}+1}=\frac{E_1E_3}{F_1F_3},
 \end{equation}
 where
 \begin{gather}
 E_1=b_1 d_1 (b_2 d_2 (1-x
 y)+c_1 d_3 y
 (c_2+x ) )+c_1
 d_2 d_3 x (c_2
 y+1 ),\nonumber\\
 E_2=b_1 d_1 (b_2 d_2 (1-x
 y)+c_3 d_3 q y
 (c_1+x ) )+c_3
 d_2 d_3 q x (c_1
 y+1 ),\nonumber\\
 E_3=b_1 d_1 (b_2 d_2 (1-x y)+c_2
 d_3 y
 (c_1+x ) )+c_3 d_2
 d_3 q x (c_1 y+1 ), \nonumber\\
 F_1=b_1 d_1 (b_2 d_2 (1-x
 y)+c_2 d_3 y
 (c_1+x ) )+c_2
 d_2 d_3 x (c_1
 y+1 ),\nonumber\\
 F_2=b_1 d_1 (b_2 d_2 (1-x y)+c_3
 d_3 q y
 (c_2+x ) )+c_3 d_2
 d_3 q x (c_2 y+1 ),\nonumber\\
 F_3= b_1 d_1 (b_2 d_2 (1-x y)+c_1
 d_3 y
 (c_2+x ) )+c_3 d_2
 d_3 q x (c_2 y+1 ),\label{fup333}
 \end{gather}
 and $\overline{*}$ stands for $T_2(*)$.
 \end{Theorem}
\begin{proof}
Solving a compatibility condition \eqref{ryrt} with \eqref{swk} for the variables $\overline{x}$, $\overline{y}$, we obtain the following equations
\[
\overline{x}=\frac{(xy-1)G_1(x,y)}{H_1(x,y)},\qquad \overline{y}=\frac{xG_2(x,y)}{H_2(x,y)},
\]
where $G_k(x,y)$, $H_k(x,y)$ $(k=1,2)$ are polynomials in variables $x$, $y$. The polynomial $G_1(x,y)$ is of degree $(1,1)$, $G_2(x,y)$ is of degree $(1,2)$ and $H_k(x,y)$ are of degree $(2,2)$.
Using a method~\cite{kny17} for finding point configuration, a configuration of points for the equation \eqref{fup33}, \eqref{fup333} is as follows:
\begin{align}
(x,y)={}& \left( -c_1, -\frac{1}{c_1}\right), \quad \left(-c_2, -\frac{1}{c_2}\right),\quad
 \left(-\frac{b_1 d_1}{d_2}, -\frac{d_2}{b_1 d_1}\right),\quad
 \left(-\frac{b_1 b_2}{c_3}, -\frac{c_3}{b_1 b_2} \right),\nonumber \\
 &\ (-b_2, 0), \quad \left(-\frac{b_1 b_2 d_1}{c_3 d_3 q}, 0\right),\quad \left(0, -\frac{1}{b_1}\right), \quad \left(0, -\frac{b_2 d_2}{c_1 c_2 d_3}\right).\label{xypoint}
\end{align}
Calculating $\frac{\overline{x}+c_1}{\overline{x}+c_2}$ and $\frac{c_1\overline{y}+1}{c_2 \overline{y}+1}$, respectively, we have equations \eqref{fup33}, \eqref{fup333}.
\end{proof}

\begin{Remark}
The time evolution equations $\overline{u_{j,i}} $ are derived by solving the compatibility condition \eqref{ryrt}.
For another derivation using the transformations which correspond to permutations of the matrices $X_i^{\varepsilon_i}(z)$, see~\cite[Proposition~2.1]{park20}.
\end{Remark}
Before ending this subsection,
we show a relation between pairs of the variables $(f,g)$ in~\eqref{e6pade} and $(x,y)$ in~\eqref{fup33}.
 \begin{Proposition} \label{fgxy}
 The equations \eqref{e6a}, \eqref{e6fg} and \eqref{mn33}, \eqref{swk}
 are equivalent with each other if the variables $f$, $g$ and $x$, $y$ are related as
 \begin{gather}
 f=\frac{c_1 c_2 c_3
 (b_2+x) (x
 y-1)}{b_1 b_2 y
 (c_1+x)
 (c_2+x)+c_3 x
 (b_2 (1-x y)+c_2 x y+c_1
 y
 (c_2+x)+x)}, \nonumber\\
 g=-\frac{y (b_1 b_2+c_3 x)}{c_3 (b_2 (1-x y)+x)+b_1b_2 x y}.\label{crsp}
 \end{gather}
 \end{Proposition}
 \begin{proof}
 Comparing the coefficient matrix $A(z)$ of the equation \eqref{e6a}, \eqref{e6fg} to the coefficient matrix $A(z)$ of the equation \eqref{mn33}, \eqref{swk}, we can solve for the variables $x$, $y$ and gauge freedom $w_1$, $w_2$ in terms of variables $f$, $g$ and parameters $b_j$, $c_i$, $d_j$ ($1\leq i, j \leq 3$).
 Then we have the desired relation between pairs of variables $(f,g)$ and $(x,y)$ \eqref{crsp}.
 \end{proof}

In Appendix \ref{rvt12}, we give pictures of a point configuration of the equation \eqref{e6pade}
and the configuration \eqref{xypoint},
root basis associated with them,
and three deformations~$T_1$,~$T_2$ and $T_2^3$ on root variables attached with root basis.

 \section[A scalar equation related to the q-E\_6\^{}(1)]{A scalar equation related to the $\boldsymbol{q}$-$\boldsymbol{E_6^{(1)}}$}\label{scle6}
In this section, we derive a scalar $q$-difference equation from the matrix $q$-difference equation \eqref{mn33}, \eqref{swk} for an unknown function $\Psi(z)=[\Psi_1(z),\Psi_2(z),\Psi_3(z)]$ and its properties.

Before deriving, we state about a characterization of a linear $q$-difference equation.
Linear differential equations are characterized by its singular points and characteristic exponents.
Similarly, linear $q$-difference equations are also characterized by its singular points and characteristic exponents.
 We consider the following $n$-th order $q$-difference equation
\begin{gather}
P_n(z) \Phi(q^nz)+P_{n-1}(z) \Phi\big(q^{n-1}z\big)+\cdots +P_0(z) \Phi(z)=0, \nonumber\\
P_k(z)= p_{k,0}+p_{k,1}z+\cdots +p_{k,n+l-k}z^{n+l-k},\qquad 0\leq k \leq n,\quad l\in \mathbb{Z}_{\geq 0},\nonumber
\\
p_{0,0}, p_{0,n+l}, p_{n,0}, p_{n,l}\neq 0.\label{neq3}
\end{gather}
In the $q$-difference equation \eqref{neq3}, singular points are at $z=0$ and $z=\infty$.
And characteristic exponents of solutions at $z=0$ and $z=\infty$ are given as solutions of the following characteristic equations respectively
\begin{gather}
P_n(0)\lambda ^n +P_{n-1}(0) \lambda ^{n-1}+\cdots + P_0(0)=0, \nonumber\\
P_n(\infty)q^{n-1}{\mu} ^n +P_{n-1}(\infty)q^{n-2} {\mu} ^{n-1}+\cdots +P_1(\infty)\mu+ P_0(\infty)=0.\label{eq3.2}
\end{gather}
To characterize the $q$-difference equation \eqref{neq3} is namely to determine coefficients $p_{k,m}$.
The number of coefficients $p_{k,m}$ in the equation \eqref{neq3} is \smash{$\frac{(n+1)(n+2l+2)}{2}$}.
Through the following, total $3n+2l-1$ coefficients $p_{k,m}$ are determined:
\begin{itemize}\itemsep=0pt
\item[(i)] We express parameters $a_i$ and $b_j$ as zeroes of the coefficients $P_n(z)$ and $P_0(z)$:
\begin{gather}
P_n(z)=p_{n,l}(z-a_1)(z-a_2)\cdots (z-a_l),\nonumber \\
P_0(z)=p_{0,n+l}(z-b_1)(z-b_2)\cdots (z-b_{n+l}).\label{factor}
\end{gather}
The parameters $a_i$ and $b_j$ indicate poles of solutions of the equation \eqref{neq3}.
By the expressions \eqref{factor}, $(2l+n)$ coefficients $p_{n,m}$, $p_{0,m'}$ ($m \neq l$, $m' \neq n+l$) are determined.
\item[(ii)] We put parameters $c_i$ and $d_j$ as characteristic exponents at $z=0$ and $z=\infty$ respectively
\begin{gather}
z=0\colon\ c_1,c_2,\dots , c_n,\nonumber \\
z=\infty\colon\ d_1,d_2,\dots, d_n.\label{exponent}
\end{gather}
There is the following relation between parameters $a_i$, $b_j$, $c_k$ and $d_l$ ($q$-Fuchs' relation)
\begin{equation}\label{qfuchs}
(-1)^n q^{\frac{n(n-1)}{2}}
\prod _{j=1}^n d_j
\prod _{i=1}^{n+l} b_i
=\prod _{k=1}^l a_k
\prod _{j=1}^n c_j.
\end{equation}
By the conditions \eqref{exponent} and \eqref{qfuchs},
solving relations between roots and coefficients of characteristic equations~\eqref{eq3.2} at $z=0$ and $z=\infty$ for the equation \eqref{neq3}, $(2n-1)$ coefficients $p_{k,n+l-k}$ $p_{k',0}$ ($k \neq n$, $k' \neq 0$) are determined.
\end{itemize}
\begin{Example}
In case $(n,l)=(2,0)$, the equation \eqref{neq3} has $6$ coefficients $p_{k,m}$.
From the conditions (i) and (ii), the equation \eqref{neq3} is characterized uniquely up to normalization.
This equation is equivalent to the $q$-hypergeometric equation via a gauge transformation.
\end{Example}
From the above (i) and (ii), the number of remain coefficients $p_{k,m}$ is $\frac{1}{2}(n-1)(n+2l-2)$, which is the number of accessary parameters.
If $z=a_1$ is an apparent singularity for the equation~\eqref{neq3},
namely all solutions of the equation \eqref{neq3} are regular at $z=a_1$,
we have the following relations:
\begin{gather}
 P_0(a_1/q)=0,\qquad \text{and}\qquad
 f:=\frac{P_0(a_1)}{P_1(a_1/q)}=\frac{P_1(a_1)}{P_2(a_1/q)}=\cdots =\frac{P_{n-1}(a_1)}{P_n(a_1/q)},\label{apparent}
\end{gather}
where $f$ is a parameter.
The above equations \eqref{apparent} correspond to a non-logarithmic condition via a Laplace transformation $z \leftrightarrow T_z$.
The relations \eqref{apparent} determine $n$ coefficients $p_{k,m}$.

From now on, we derive a scalar $q$-difference equation from the matrix $q$-difference equation \eqref{mn33}, \eqref{swk} for an unknown function $\Psi(z)=[\Psi_1(z),\Psi_2(z),\Psi_3(z)]$ and its properties.
Eliminating functions $\Psi_2(z)$ and $\Psi_3(z)$ in the equation \eqref{mn33}, \eqref{swk}, we obtain the following third linear $q$-difference equation for $\Phi(z):=\Psi_1(z)$:
 \begin{equation}\label{scalar}
 L(z):=P_3(z)\Phi\big(q^3 z\big)+P_2(z) \Phi\big(q^2 z\big)+P_1(z)\Phi(q z)+P_0(z)\Phi(z)=0,
 \end{equation}
 where
 \begin{gather}
P_3(z)=p_{31}(z-u),\nonumber\\
P_2(z)=p_{22}z^2+p_{21}z+p_{20},\nonumber\\
P_1(z)=p_{13}z^3+p_{12}z^2+p_{11}z+p_{10},\nonumber\\
P_0(z)=-P_3(qz)d_1d_2d_3(z+c_1)(z+c_2)(z+c_3).\label{pjz}
\end{gather}
 Here, the coefficients $p_{k,l}$ $(1\leq k \leq 3,\, 0 \leq l\leq 3)$ in the polynomials $P_{k}(z)$ \eqref{pjz} depend on parameters $b_j$, $c_i$, $d_j$, and a variable $u$ defined as the zero of $P_3(z)$.
 The variable $u$ is expressed in terms of $x$, $y$ as follows
 \begin{equation}\label{udef}
 u=\frac{I_1(x,y)I_2(x,y)}{J_1(x,y)J_2(x,y)},
 \end{equation}
 where
 \begin{gather*}
 I_1(x,y)=y (b_1 (c_2 x
 y+c_1 y
 (c_2+x)+x)
 +c _1 c_2 (1-x y)),\\
 I_2(x,y)=
 b_2 d_3
 I_1(x,y)-b_2
 d_2
 J_1(x,y)
 +c_3 d_3 x (c_1
 y+1) (c_2
 y+1),\\
 J_1(x,y)=b_2
 (b_1 y+1) (x y-1),
 \\ J_2(x,y)=
 d_2 (
 J_1(x,y)
 -x(c_1 y+1)
 (c_2
 y+1))-d_3
 I_1(x,y).
 \end{gather*}
 Explicit forms of the polynomials $P_j(z)$ ($0\leq j \leq 3$) \eqref{pjz} are given in appendix.

Then we have
\begin{Lemma}
The equation $L(z)=0$ \eqref{scalar} has the following properties:
\begin{enumerate}\itemsep=0pt
\item[$(i)$] it is a linear four term equation between $\Phi\big(q^jz\big)$ $(0 \leq j\leq 3)$ and its coefficients $P_{j}(z)$ are polynomials for $z$ of degree $4-j$,
\item[$(ii)$] a polynomial $P_0(z)$ has four zero points at $z=-c_i$ $(1 \leq i \leq 3)$, $u/q$,
\item[$(iii)$] the exponents of solutions $\Phi(z)$ are $b_1d_1$, $qb_2 d_2$, $qb_3d_3$ $($at $z=0)$ and $d_1$, $d_2$, $d_3$ $($at~$z=\infty)$,
\item[$(iv)$] a point $z=u$ such that $P_3(z)=0$ is an apparent singularity,
namely we have
\begin{equation}\label{gdef}
v:=\frac{P_0(u)}{P_1(u/q)}=\frac{P_1(u)}{P_2(u/q)}=\frac{P_2(u)}{P_3(u/q)}.
\end{equation}
\end{enumerate}
Conversely, the equation $L(z)=0$ \eqref{scalar} is uniquely characterized by these properties $(i)$--$(iv)$ up to normalization.
\end{Lemma}
\begin{proof}
The properties (i)--(iv) follows by computation through eliminating $\Psi_2(z)$, $\Psi_3(z)$ in \eqref{mn33}, \eqref{swk}.
The converse can be confirmed that coefficients $P_j(z)$ are defined uniquely by (i)--(iv) up to a normalization.
To see this, we consider the following equation which satisfies the properties (i), (ii):
\begin{gather}
L'(z)=P_3'(z)\Phi\big(q^3 z\big)+P_2'(z) \Phi\big(q^2 z\big)+P_1'(z)\Phi(q z)+P_0'(z)\Phi(z)=0,\label{lds}\\
P_3'(z)=p_{31}'(z-u),\nonumber\\
P_2'(z)=p_{22}'z^2+p_{21}'z+p_{20}',\nonumber\\
P_1'(z)=p_{13}'z^3+p_{12}'z^2+p_{11}'z+p_{10}',\nonumber\\
P_0'(z)=p_{04}'(z-u/q)(z+c_1)(z+c_2)(z+c_3).\nonumber
\end{gather}
From the property (iii), the condition of the exponents of solutions $\Phi(z)$ at $z=0$ determines the coefficients $p_{31}'$, $p_{20}'$, $p_{10}'$
and the condition of the exponents of solutions $\Phi(z)$ at $z=\infty$ determines the coefficients $p_{22}'$, $p_{13}'$.
The remaining coefficients except for $p_{04}'$ are determined by the property~(iv).
If we put the normalization factor $p_{04}'$ as $q^5 uvc_1c_2d_1d_2d_3$, the function~$L'(z)$~\eqref{lds} equals to the function $L(z)$ \eqref{scalar}.
\end{proof}

In the following, viewing $\Phi\big(q^iz\big)$ $(0 \leq i \leq 3)$ as parameters, we regard the scalar $q$-difference equation $L(z)=0$ \eqref{scalar} as an algebraic curve in variables $(u,v)$ $\in$ $\mathbb{P}^1\times \mathbb{P}^1$.
 We represent as the curve as $P(u,v)=0$.
The features of the curve are the following.
\begin{Lemma} The algebraic curve $P(u,v)=0$ has the following properties:
\begin{enumerate}\itemsep=0pt
\item[$(i)$] The polynomial $P(u,v)$ has the following form:
\[
P(u,v)=\displaystyle \sum_{0\leq j \leq 3\atop 0\leq i \leq 4-j}c_{i,j} u^iv^j, \qquad c_{0,0}:=c_0 z^2 \Phi(q z).
\]
The coefficients $c_{i,j}$ depend on $b_i$, $c_i$, $d_i$, $q$, $z$, $\Phi\big(q^iz\big)$.
\item[$(ii)$] It passes the following $8$ points:
\begin{gather*}
(u,v)=(0,qb_1d_1), \quad \big(0,q^2b_2d_2\big), \quad \big(0,q^2b_3d_3\big), \quad (qz,\infty), \\
\phantom{(u,v)=}{} \ (z,0), \quad(-c_1,0),\quad (-c_2,0),\qquad (-c_3,0),
\end{gather*}
and $3$ points in the coordinate $(r,s)=(u,v/u)$
\[
(r,s)=\big(\infty,q^2 d_1\big),\quad \big(\infty,q^2 d_2\big),\quad \big(\infty,q^2 d_3\big).
\]
\item[$(iii)$]
At $u=z$ the equation $P(u,v)=0$ has the following property:
\begin{equation}\label{prpt}
\displaystyle \bigg( \sum _{0 \leq i \leq 3-j} c_{i,j+1}z^i \bigg)\Big|_{z\to q z, \Phi(q^kz)\to \Phi(q^{k+1}z)}
=\sum _{0 \leq i \leq 4-j} c_{i,j}(qz)^i, \qquad j=0, 1, 2.
\end{equation}
\end{enumerate}
Conversely, the equation $P(u,v)=0$ is uniquely characterized by these properties $(i)$--$(iii)$ up to normalization factor $c_0$.
\end{Lemma}
\begin{proof}
The properties (i)--(iii) follow for the polynomial $P(u,v)$.
Conversely, we consider a~polynomial \[P'(u,v):= \sum_{0\leq j \leq 3\atop 0\leq i \leq 4-j}c_{i,j}'u^iv^j , \qquad c_{0,0}'(z)=c_0' \Phi(qz).\]
The polynomial $P'(u,v)$ has $14$ coefficients.
From the property (ii), $10$ coefficients are described in terms of parameters $b_j$, $c_j$, $d_j$ and the coefficient $c_0'$ as follows:
\begin{gather*}
 P'(u,v)= c_0' z \Phi(q z) \frac{
 (u-qz) }{b_1 b_2
 b_3 d_1 d_2 d_3 q^6}v^3 \\
\hphantom{P'(u,v)=}{} + \bigg(c_{1,2}'u-c_0'
 \frac{z\Phi(q z)}{b_1
 b_2 b_3 d_1 d_2 d_3 q^4} \big((d_1
 +d_2 +d_3 )u^2(b_2 d_2 q
 +b_3 d_3 q +b_1 d_1 )z\big)\bigg)v^2\\
\hphantom{P'(u,v)=}{}+\bigg( c_{2,1}'u^2+c_{1,1}'u+c_0' z\Phi(q z)
\bigg( \bigg( \frac{1}{d_1}+\frac{1}{d_2}+\frac{1}{d_3} \bigg)
 \frac{u^3}{q^2b_1b_2b_3}\\
\hphantom{P'(u,v)=}{}-z \bigg( \frac{1}{b_1d_1q}+\frac{1}{b_2d_2q^2}+ \frac{1}{b_3d_3q^2} \bigg) \bigg)\bigg) v \\
\hphantom{P'(u,v)=}{}+c_0'z \Phi(q z)\frac{
 (c_1+u )
 (c_2+u )
 (c_3+
 u)
 (z-u)
 }{b_1 b_2
 b_3 }.
\end{gather*}
The remaining $3$ parameters $c_{1,1}'$, $c_{1,2}'$, $c_{2,1}'$ are determined in terms of parameters $b_i$, $c_i$, $d_i$, $q$, $z$, $\Phi\big(q^iz\big)$, $c_0'$
by the property (iii).
Namely, the property (iii) gives $3$ linear inhomogeneous equations among $c_{1,1}'$, $c_{1,2}'$, $c_{2,1}'$, $c_{1,1}'\big|_{z\to qz}$, $c_{1,2}'\big|_{z\to qz}$, $c_{2,1}'\big|_{z\to qz}$.
Though these relations are apparently $q$-difference equations, we can solve them algebraically.
For example, in the equation \eqref{prpt}, we solve $c_{12}'$ when $j=2$. Then when $j=0$, we solve $\left. c_{21}' \right| _{z \to qz, \Phi(q^kz)\to \Phi (q^{k+1}z)}$.
And finally solving~$c_{11}'$ when $j=1$, they algebraically can be solved.
\end{proof}

 Explicit forms of the coefficients $c_{i,j}(z)$ of the polynomial $P(u,v)$ are in Appendix \ref{b}.
\subsection[Relations among pairs of variables (f, g), (x, y) and (u, v)]{Relations among pairs of variables $\boldsymbol{(f, g)}$, $\boldsymbol{(x, y)}$ and $\boldsymbol{(u, v)}$} \label{taiou}
 \begin{Proposition}
Under the relations among variables $(f, g)$, $(x, y)$ and $(u, v)$:
\begin{gather}
f=\frac{c_1 c_2 c_3
 (b_2+x) (xy-1)}{b_1 b_2 y
 (c_1+x)
 (c_2+x)+c_3 x
 (b_2 (1-x y)+c_2 x y+c_1
 y
 (c_2+x)+x)},\nonumber \\
 g=-\frac{y (b_1 b_2+c_3 x)}{c_3 (b_2 (1-x y)+x)+b_1b_2 x y},\label{f2x}\\
x=\frac{K_1(u,v)K_2(u,v)}{L_1(u,v)}, \qquad
y=\frac{K_3(u,v)}{L_2(u,v)},\label{x2u}\\
K_1(u,v)=b_2 v \big(v-q^2 (b_2 d _2+ d _3 u)\big),\nonumber \\
K_2(u,v)=b_1 b_2\big(q^2 \big( d _3 u ( d _2 q^2 (c_3+u)-v)- d _2 v(b_2+u)\big)+v^2\big)\nonumber\\
\phantom{K_2(u,v)=}{} +c_1 c_2 d _3 q^2 \big(b_2 d _2 q^2(c_3+u)-c_3 v\big),\nonumber\\
K_3(u,v)=b_2 d _2 q^2 u \big( d _3 q^2 (c_3+u)-v\big),\nonumber\\
L_1(u,v)=q^2 \big(b_2^2 d _2
\big( d _2 q^2 \big( d _3 q^2 (c_1+u)
 (c_3+u)-u v\big) \big( d _3 q^2 (c_2+u)(c_3+u)-u v\big)\nonumber\\
 \phantom{L_1(u,v)=}{} +v \big( d _3 q^2 v (b_1 (2c_3+u)+u (c_3+2 u)) - d _3^2 q^4 u(b_1+u) (c_3+u)-u v^2\big)\big)\nonumber\\
 \phantom{L_1(u,v)=}{} -b_2^3 d _2^2q^2 v ( d _3 q^2 (b_1+u)(c_3+u )-u v)-b_2 c_3 d _3 v \big(b_1 v \big(v- d _3 q^2u\big)\nonumber\\
 \phantom{L_1(u,v)=}{} + d _2 q^2 \big( d _3 q^2 (c_2 u+c_1 (2c_2+u)) (c_3+u)-(c_1+c_2) uv\big)\big)+c_1 c_2 c_3^2 d _3^2 q^2 v^2\big),\nonumber\\
 L_2(u,v)=c_1 c_2 d _3 q^2 \big(b_2 d _2 q^2 (c_3+u)-c_3 v\big)+b_1 b_2
 v \big(v-q^2 (b_2 d _2+ d _3 u)\big),\nonumber
\end{gather}
the corresponding Lax equations \eqref{e6a} with \eqref{v2f}, \eqref{mn33} with \eqref{swk} and \eqref{mn33} with \eqref{udef}, \eqref{gdef} are equivalent.
Conversely, such relations among $(f,g)$, $(x,y)$ and $(u,v)$ are uniquely determined as \eqref{f2x}, \eqref{x2u}.
 \end{Proposition}
\begin{proof}
Using the relation \eqref{f2x}, we can check that the equation \eqref{e6a} with \eqref{v2f} is equivalent to \eqref{mn33} with \eqref{swk}.
Similarly, using the relation \eqref{x2u}, we can check that the equation \eqref{mn33} with \eqref{swk} is equivalent to \eqref{mn33} with \eqref{udef}, \eqref{gdef}.
The converse is obvious from the form of the Lax matrix.
\end{proof}

\section{Continuous limit} \label{blc}
In this section, we describe a relation between our result and the result of Boalch \cite{boalch09}.
In \cite{boalch09}, a Lax pair for the additional-difference Painlev\'e equation with affine Weyl symmetry group of type $E_6$ was described.
The linear differential equation of the Lax pair is as follows
\begin{equation}\label{e6fuchs}
\frac{d}{dz}\Psi(z)=\Psi(z) \left( \frac{A_1^b}{z}+\frac{A_2^b}{z-1} \right) ,\qquad A_3^b:=-\big(A_1^b+A_2^b\big),
\end{equation}
where the matrices $A_i^b$ ($1\leq i \leq 3$) are $3 \times 3$ matrices with different eigenvalues.

We show that the linear $q$-difference equation \eqref{lax}
reduces to the equation \eqref{e6fuchs} via a continuous limit $q\to 1$.
The equation \eqref{lax} takes the following form after a scale transformation~$z \to -z$ and gauge transformations
 \begin{gather}
 (1-z)\Psi(qz)=\Psi(z) \mathcal{A}(z) , \qquad \mathcal{A}(z) =\begin{bmatrix}
 b_1d_1 & k_1 & v_1 \\
 0 & b_2d_2 & k_2 \\
 0 & 0 & b_3d_3
\end{bmatrix}
-\begin{bmatrix}
 d_1 & 0 & 0 \\
 v_2 & d_2 & 0 \\
 v_3 & v_4 & d_3
 \end{bmatrix}z, \nonumber\\
|\mathcal{A}(z)|=d_1d_2d_3 (z-c_1)(z-c_2)(z-c_3),\label{e6a2}
 \end{gather}
 where $k_j$ ($j=1$, $2$) are constants.
We put $q=e^h$ and consider the limit $h \to 0$.
We set
\begin{gather}
b_i=q^{\beta_i},\qquad c_i=q^{\gamma_i},\qquad d_i=q^{\delta_i},\qquad 1\leq i \leq 3,\nonumber \\
k_j=h l_j, \qquad j=1,2, \qquad
v_m=hu_m,\qquad 1\leq m \leq 4,\label{q21}
\end{gather}
where $l_j$ are constants.
 By using Taylor's expansion for \eqref{e6a2}, \eqref{q21}
\begin{gather*}
\text{(l.h.s.)}= (1-z) \Psi(z)+h(1-z)z \frac{\rm d}{{\rm d}z} \Psi(z)+\mathcal{O}\big(h^2\big),
\\
\text{(r.h.s.)}= (1-z) \Psi(z) \\
\hphantom{\text{(r.h.s.)}=}{} +h
\Psi(z)
\begin{bmatrix}
\beta _1-\delta _1 (z-1) & l_1
 & u_1 \\
 -u_2 z & \beta _2-\delta _2
 (z-1) & l_2 \\
 -u_3 z & -u_4 z & \beta
 _3-\delta _3 (z-1)
 \end{bmatrix}
 +\mathcal{O}\big(h^2\big),
\end{gather*}
we find the following limit as $h \to 0$:
\begin{gather}
\frac{\rm d}{{\rm d}z} \Psi(z)=\Psi(z) \left( \frac{A_1}{z}+\frac{A_2}{z-1} \right) ,\nonumber \\
A_1=\begin{bmatrix}
\beta _1+\delta _1 & l_1 & u_1 \\
 0 & \beta _2+\delta _2 & l_2 \\
 0 & 0 & \beta _3+\delta _3
 \end{bmatrix},\qquad
A_2=\begin{bmatrix}
-\beta _1 & -l_1 & -u_1 \\
 u_2 & -\beta _2 & -l_2 \\
 u_3 & u_4 & -\beta _3
\end{bmatrix},\nonumber \\
A_3:=-(A_1+A_2)
=\begin{bmatrix}
\delta _1 & 0 & 0 \\
 u_2 & \delta _2 & 0 \\
 u_3 & u_4 & \delta _3
\end{bmatrix},\label{q21e}
\end{gather}
where eigenvalues of the matrix $A_2$ are $\gamma_i$ ($1 \leq i \leq 3$) by the condition of the determinant of the matrix $\mathcal{A}(z)$ \eqref{e6a2}.
Therefore, the linear $q$-difference equation~\eqref{lax} reduces to the equation~\eqref{e6fuchs} via a continuous limit $q\to 1$.

In the following, we consider a continuous limit $q \to 1$ of the result in Section~\ref{pd}.
In Section~\ref{pd},
through a compatibility condition of the equations \eqref{ryrt1}, we derived a standard $q$-Painlev\'e equation of type $E_6$.
We take the following equation as a deformation equation for the differential equation \eqref{q21e} which is rewritten version of \eqref{deq1}:
\begin{gather}
T\Psi(z)=B(z)\Psi(z),\qquad
 B(z)=
 \begin{bmatrix}
 w_1 & w_2 & w_3 \\
 0 & 0 & 0 \\
 0 & 0 & 0
 \end{bmatrix}+
 \begin{bmatrix}
 w_4 & 0 & 0 \\
 1 & 0 & 0 \\
 w_5 & w_6 & w_7
 \end{bmatrix}z, \nonumber\\
 |B(z)|= (w_3 w_6-w_2 w_7) z^2, \nonumber\\
T\colon\ (\beta_1,\beta_2,\beta_3,\gamma_1,\gamma_2,\gamma_3, \delta_1,\delta_2, \delta_3)
 \to (\beta_1-1, \beta_2+1, \beta_3,\gamma_1,\gamma_2,\gamma_3, \delta_1+1,\delta_2,\delta_3+1).\label{sch}
\end{gather}
Solving a compatibility condition for the equation \eqref{q21e}, \eqref{sch}, we obtain the following additional-difference Painlev\'e equation of type $E_6$ \cite{kny17,rgtt01}:
\begin{gather*}
(f+g)\big(\overline{f}+g\big)=\frac{ (g+\gamma _1 ) (g+\gamma _2 ) (g+\gamma_3 ) (g+\beta
 _1+\delta _1-\delta _2 )}{ (g+\beta _2+1 )
 (g+\beta _3-\delta _2+\delta _3+1 )},\\
 \big(\overline{f}+\overline{g}\big)\big(\overline{f}+g\big)=\frac{\big(\overline{f}-\gamma_1\big)\big(\overline{f}-\gamma_2\big)\big(\overline{f}-\gamma_3\big)
 \big(\overline{f}-\beta_1-\delta_1+\delta_2\big)}{\big(\overline{f}-\beta_1+\beta_2-1\big)\big(\overline{f}-\beta_1-\delta_1+\delta_3+1\big)},
\end{gather*}
where
\begin{gather*}
f=\beta _3+\frac{u_1 u_4}{l_1}, \qquad
g=\frac{l_1 l_2}{u_1}-\beta _2,
\end{gather*}
and $\overline{*}$ stands for $T(*)$.
From the above, we derive the additional-difference Painlev\'e equation of type $E_6$ solving a compatibility condition of the Lax pair via a continuous limit $q \to 1$.
\appendix

\section[Deformations T\_1, T\_2 and T\_2\^{}3 on root variables]{Deformations $\boldsymbol{T_1}$, $\boldsymbol{T_2}$ and $\boldsymbol{T_2^3}$ on root variables} \label{rvt12}
 In this appendix,
we show that how the deformations $T_1$, $T_2$ and $T_2^3$ act on root variables.

We consider a pair of root basis of $\{ \alpha_i\}$ ($i=0,1,\dots, 6$) and $\{ \delta_j \}$ ($j=0,1,2$) as symmetry type \smash{$E_6^{(1)}$} and surface type \smash{$A_2^{(1)}$}, respectively.
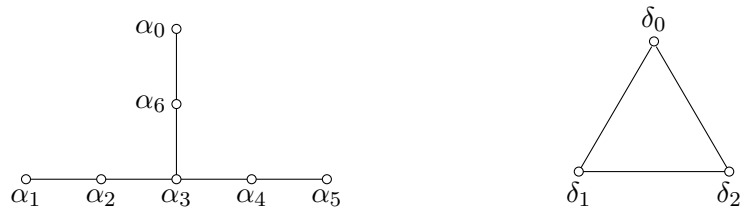
\begin{figure}[h] \hspace{80pt}
\begin{minipage}[t]{300pt}
\begin{tikzpicture}
\draw(0,0)circle(0.06)node[below]{$ \alpha_3$};
	\draw(-1,0)circle(0.06)node[below]{$\alpha_2$};
	\draw(-2,0)circle(0.06)node[below]{$\alpha_1$};
	\draw(1,0)circle(0.06)node[below]{$\alpha_4$};
	\draw(2,0)circle(0.06)node[below]{$\alpha_5$};
	\draw(0,1)circle(0.06)node[left]{$\alpha_6$};
	\draw(0,2)circle(0.06)node[left]{$\alpha_0$};
	\draw(0-0.06,0)--(-1+0.06,0);
	\draw(-2+0.06,0)--(-1-0.06,0);
	\draw(0+0.06,0)--(1-0.06,0);
	\draw(2-0.06,0)--(1+0.06,0);
	\draw(0,0+0.06)--(0,1-0.06);
	\draw(0,1+0.06)--(0,2-0.06);
\end{tikzpicture}
\hspace{70pt}
\begin{tikzpicture}
\draw(0,0)circle(0.06)node[below]{$\delta_1$};
\draw(2,0)circle(0.06)node[below]{$\delta_2$};
\draw(1,{sqrt(3)})circle(0.06)node[above]{$\delta_0$};
\draw(0.03,{0.03*sqrt(3)})--(1-0.03,{0.96*sqrt(3)});
\draw(0.06,0)--(2-0.06,0);
\draw(1+0.03,{0.96*sqrt(3)})--(2-0.03,{0.03*sqrt(3)});
\end{tikzpicture}
\caption{The Dynkin diagram of $E_6^{(1)}$ and $A_2^{(1)}$.}
\end{minipage}
\end{figure}

 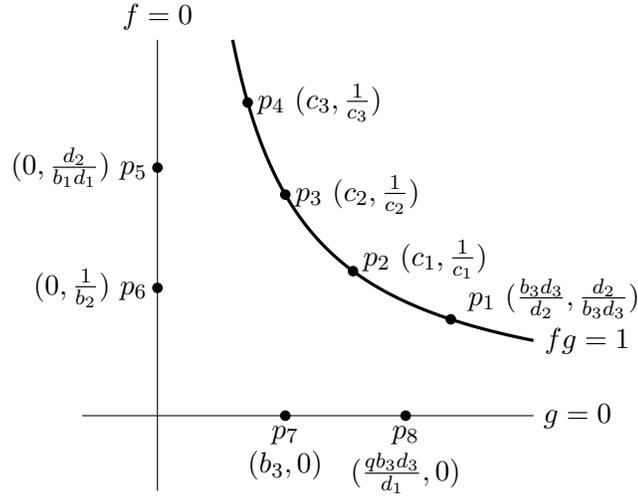
\begin{figure}[t]
\centering
\begin{tikzpicture}
 \path[draw] (-1.0,0) -- (5.0,0)node[right]{$g=0$};
 \path[draw] (0,-1.0) -- (0,5.0)node[above]{$f=0$};
 \draw[very thick, samples=100, domain=1:5]plot(\x,{5/(\x)})node[right]{$fg=1$};
 \fill(1.2, 5/1.2)circle[radius=2pt];
 \draw(1.2, 5/1.2)node[right]{$p_4$ $(c_3,\frac{1}{c_3})$};
 \fill(1.7, 5/1.7)circle[radius=2pt];
 \draw(1.7, 5/1.7)node[right]{$p_3$ $(c_2,\frac{1}{c_2})$};
 \fill(2.6, 5/2.6)circle[radius=2pt];
 \draw(2.6, 5/2.4)node[right]{$p_2$ $(c_1,\frac{1}{c_1})$};
 \fill(3.9, 5/3.9)circle[radius=2pt];
 \draw(4, 5/3.2)node[right]{$p_1$ ($\frac{b_3d_3}{d_2},\frac{d_2}{b_3d_3}$)};
 \fill(0, 1.7)circle[radius=2pt];
 \draw(0, 1.7)node[left]{$(0,\frac{1}{b_2})$ $p_6$};
 \fill(0, 3.3)circle[radius=2pt];
 \draw(0, 3.3)node[left]{$(0,\frac{d_2}{b_1d_1})$ $p_5$};
 \fill(1.7,0)circle[radius=2pt];
 \draw(1.7,0)node[below]{$p_7$};
 \draw(1.7,-1/3)node[below]{$(b_3,0)$};
 \fill(3.3,0)circle[radius=2pt];
 \draw(3.3,0)node[below]{$p_8$};
 \draw(3.3,-1/2.8)node[below]{$(\frac{qb_3d_3}{d_1},0)$};
 \end{tikzpicture}
 \caption{A point configuration of the equation \eqref{e6pade}.}
 \label{fgconf}
 \end{figure}

 \begin{figure}[t]\centering
 \begin{tikzpicture}
 \path[draw] (-1.0,0) -- (5.0,0)node[right]{$y=0$};
 \path[draw] (0,-1.0) -- (0,5.0)node[above]{$x=0$};
 \draw[very thick, samples=100, domain=1:5]plot(\x,{5/(\x)})node[right]{$xy=1$};
 \fill(1.2, 5/1.2)circle[radius=2pt];
 \draw(1.2, 5/1.2)node[right]{$q_4$ $(-c_{2},-\frac{1}{c_{2}})$};
 \fill(1.7, 5/1.7)circle[radius=2pt];
 \draw(1.7, 5/1.7)node[right]{$q_3$ $(-c_{1},-\frac{1}{c_{1}})$};
 \fill(2.6, 5/2.6)circle[radius=2pt];
 \draw(2.6, 5/2.3)node[right]{$q_2$ $(-\frac{b_1d_1}{d_2},-\frac{d_2}{b_1d_1})$};
 \fill(3.9, 5/3.9)circle[radius=2pt];
 \draw(4, 5/3.3)node[right]{$q_1$ $(-\frac{b_1b_2}{c_3},-\frac{c_3}{b_1b_2})$};
 \fill(0, 1.7)circle[radius=2pt];
 \draw(0, 1.7)node[left]{$(0,-\frac{1}{b_1})$ $q_6$};
 \fill(0, 3.3)circle[radius=2pt];
 \draw(0, 3.3)node[left]{$(0,-\frac{b_2d_2}{c_1c_2d_3})$ $q_5$};
 \fill(1.7,0)circle[radius=2pt];
 \draw(1.7,0)node[below]{$q_7$};
 \draw(1.7,-1/3)node[below]{$(-b_2,0)$};
 \fill(3.3,0)circle[radius=2pt];
 \draw(3.3,0)node[below]{$q_8$};
 \draw(3.7,-1/3)node[below]{$(-\frac{b_1b_2d_1}{qc_3d_3},0)$};
 \end{tikzpicture}
 \caption{A point configuration \eqref{xypoint}.}\label{xyconf}
\end{figure}
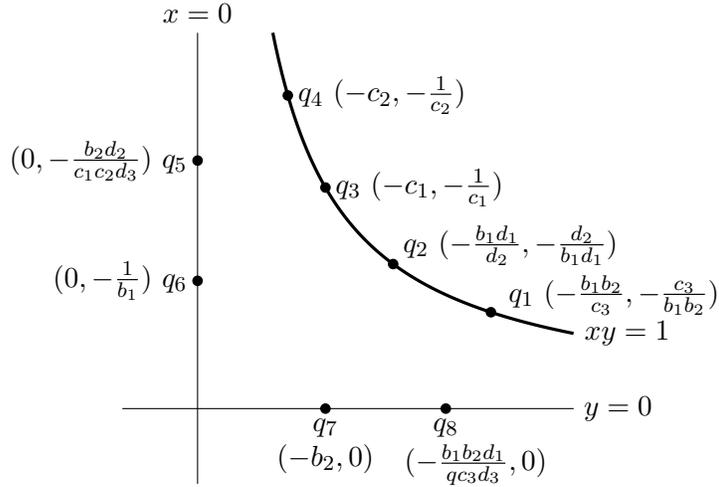

Pictures of a point configuration of the equation \eqref{e6pade} and the configuration \eqref{xypoint} are presented in Figures~\ref{fgconf} and \ref{xyconf}, respectively.
From these pictures, we take a pair of root basis $\{ \alpha_i\}$ and $\{ \delta_j\}$ as symmetry type $E_6^{(1)}$ and surface type $A_2^{(1)}$ as follows:
\begin{gather}
\alpha_0=E_7-E_8,\qquad \alpha_1=E_6-E_5, \qquad\alpha_2=H_2-E_1-E_6, \qquad \alpha_3=E_1-E_2,\nonumber\\
\alpha_4=E_2-E_3,\qquad \alpha_5=E_3-E_4,\qquad \alpha_6=H_1-E_1-E_7,\nonumber\\
\delta_0=H_1+H_2-E_1-E_2-E_3-E_4,\qquad
\delta_1=H_1-E_5-E_6,\qquad
\delta_2=H_2-E_7-E_8, \nonumber\\
\delta=\alpha_0+\alpha_1+2\alpha _2+3 \alpha_3 +2 \alpha_4 + \alpha_5+2 \alpha_6=\delta_0+\delta_1+\delta_2,\label{root}
\end{gather}
where we stand for $\delta$ as a null root.
The choice of the above root basis is the same as in \cite{kny17}.
\subsection[A deformation T\_1 on root variables]{A deformation $\boldsymbol{T_1}$ on root variables}
 We take variables $a_i$ ($i=0,1,\dots,6$) as root variables attached to the root $\alpha_i$ \eqref{root} associated with a point configuration in coordinate $(f,g)$ (see Figure~\ref{fgconf})
 \begin{gather}
a_0=\frac{qd_3}{d_1},\qquad a_1=\frac{b_1d_1}{b_2d_2},\qquad
a_2=\frac{b_2d_2}{b_3d_3},\qquad
a_3=\frac{b_3d_3}{c_1d_2},\nonumber\\
a_4=\frac{c_1}{c_2},\qquad
a_5=\frac{c_2}{c_3},\qquad
a_6=\frac{d_2}{d_3},\label{rv1}
\end{gather}
which satisfy $q=a_0a_1a_2^2a_3^3a_4^2a_5a_6^2$.
Then we have the following statement.

\begin{Proposition} The action $T_1$ \eqref{deq1} on the root variables $a_i$ in \eqref{rv1} is given by the translation
\begin{equation}\label{t1root}
T_1(a_0,a_1,a_2,a_3,a_4,a_5,a_6)=
(a_0,a_1,qa_2,a_3,a_4,a_5,a_6/q).
\end{equation}
\end{Proposition}
\begin{proof}Applying \eqref{deq1}, we obtain the desired result \eqref{t1root}.
\end{proof}

\subsection[Deformations T\_2 and T\_2\^{}3 on root variables]{Deformations $\boldsymbol{T_2}$ and $\boldsymbol{T_2^3}$ on root variables}
The deformation $T_2$ \eqref{deq2} is not a translation on parameters $b_i$, $c_i$, $d_j$ but a deformation $T_2^3$ gives a translation on them
\begin{gather}\label{t2to3}
 T_2^3\colon\ (b_1,b_2,b_3,c_1,c_2,c_3,d_1,d_2,d_3) \to \left(qb_1,qb_2,b_3,c_1,c_2,q^3 c_3,\frac{d_1}{q},\frac{d_2}{q},\frac{d_3}{q}\right).
\end{gather}
We show that how the deformations $T_2$ \eqref{deq2} and $T_2^3$ \eqref{t2to3} on root variables.

We take variables $a_i'$ ($i=0,1,\dots, 6$) as root variables attached to the root $\alpha_i$ \eqref{root} associated with a point configuration in coordinate $(x,y)$ (see Figure~\ref{xyconf})
\begin{gather}
a_0'=\frac{qc_3d_3}{b_1d_1},\qquad
a_1'=\frac{b_1b_2d_2}{c_1c_2d_3},\qquad
a_2'=\frac{b_2}{c_3},\qquad
a_3'=\frac{c_3d_1}{b_2d_2},\nonumber\\
a_4'=\frac{c_1d_2}{b_1d_1},\qquad
a_5'=\frac{c_2}{c_1},\qquad
a_6'=\frac{b_1}{c_3},\label{rv2}
\end{gather}
which satisfy $q=a_0'a_1'{a_2'}^2{a_3'}^3{a_4'}^2{a_5'}{a_6'}^2$.
Then we have the following statement.

\begin{Proposition}The actions $T_2$ \eqref{deq2} and $T_2^3$ on the root variables $a_i'$ in \eqref{rv2} are given as follows:
\begin{gather}
T_2(a_0',a_1',a_2',a_3',a_4',a_5',a_6')
=\left(\frac{q}{a_6'},{a_0'}^2a_1'a_2'a_3'{a_6'}^2,
\frac{1}{a_0'a_3'a_6'},\frac{q}{a_2'	},\frac{a_0'a_2'a_3'a_4'a_6'}{q},a_5',\frac{a_2'a_3'a_6'}{q}\right),\nonumber \\
T_2^3 (a_0',a_1',a_2',a_3',a_4',a_5',a_6')
=\left(a_0'q^2,a_1'q^3,
\frac{a_2'}{q^2},a_3'q^2,\frac{a_4'}{q},a_5',\frac{a_6'}{q^2}\right).\label{t2root}
\end{gather}
\end{Proposition}
\begin{proof}
Applying \eqref{deq2} and \eqref{t2to3}, we obtain the desired results \eqref{t2root}.
\end{proof}

\section{Explicit forms of coefficients in Section~\ref{scle6}}\label{b}

In this appendix, we give explicit forms of $P_j(z)$ \eqref{pjz} and the coefficients $c_{ij}$ $(0\leq j \leq 3,\allowbreak 0\leq i \leq 4-j)$ of the polynomial $P(u,v)$ in variables $u$ and $v$.

Explicit forms of $P_j(z)$ \eqref{pjz} are as follows:
\begin{gather*}
P_3(z)=p_{31}(z-u),\\
P_2(z)=p_{22}z^2+p_{21}z+p_{20},\\
P_1(z)=p_{13}z^3+p_{12}z^2+p_{11}z+p_{10},\\
P_0(z)=-P_3(qz)d_1d_2d_3(z+c_1)(z+c_2)(z+c_3),
\end{gather*}
where the coefficients $p_{j,k}$ $(1 \leq j \leq 3,\, 0 \leq k\leq 3)$ are
\begin{gather*}
p_{31}= -q^2 u
 v c_1 c_2,\\
 p_{22}= q^4 u v c_1
 c_2
 (d_1+d_2+d_3 ), \\
 p_{21}= -q u v c_1 c_2 \big(u d_3 q^3+ (qu-b_2 ) d_2 q^2-b_3 d_3 q^2-vq+v+ \big(q^3 u-q b_1 \big)d_1\big), \\
 p_{20}= -q^2 u^2 v c_1 c_2 (b_1 d_1+q
 b_2 d_2+q b_3 d_3 ), \\
 p_{13}= -q^5 u v c_1 c_2
 (d_1d_2+d_2 d_3+d_3d_1
 ), \\
 p_{12}= -q \big(u b_1 b_2 b_3
 (u+c_2 ) d_1 d_2 d_3 q^5+u
 c_1^2 c_2 (u+c_2 ) d_1 d_2
 d_3 q^5 \\
\hphantom{p_{12}=}{} +c_1 \big(u^2 c_2^2 d_1 d_2
 d_3 q^5+u b_1 b_2 b_3 d_1 d_2 d_3
 q^5 \\
 \hphantom{p_{12}=}{} -c_2 \big(q d_1 \big(q u \big(q
 u d_2 \big(-u d_3 q^2+v q+v \big)+v
 (q (q+1) u
 d_3-v ) \big) \\
\hphantom{p_{12}=}{} -b_1 \big(v-q^2
 b_2 d_2 \big) \big(v-q^2 b_3
 d_3 \big)\big)+v \big(d_2 \big(u
 (q (q+1) u d_3-v ) \\
\hphantom{p_{12}=}{} +b_2
 \big(q^2 b_3 d_3-v \big)\big)
 q^2+v \big(v-q^2 (u+b_3 )
 d_3 \big)\big)\big)\big)\big),\\
 p_{11}= u
 \big(u b_1 b_2 b_3 (u+c_2 )
 d_1 d_2 d_3 q^6+u c_1^2 c_2
 (u+c_2 ) d_1 d_2 d_3 q^6\\
\hphantom{p_{11}=}{}+c_1
 \big(u^2 c_2^2 d_1 d_2 d_3 q^6+u b_1
 b_2 b_3 d_1 d_2 d_3 q^6 \\
\hphantom{p_{11}=}{} -c_2 \big(q
 d_1 \big(b_1 \big(b_2 d_2 \big(-b_3
 d_3 q^3+v q+v \big) q^2+v \big(q^2
 (q+1) b_3 d_3-v \big) \big) \\
\hphantom{p_{11}=}{} -q u
 \big(v-q^2 u d_2 \big) \big(v-q^2 u
 d_3 \big)\big)+v \big(d_2 \big(u
 \big(q^2 u d_3-v \big)\\
 \hphantom{p_{11}=}{} +b_2 \big(q^2
 (q+1) b_3 d_3-v \big)\big) q^2
 +v
 \big(v-q^2 (u+b_3 )
 d_3 \big)\big)\big)\big)\big),\\
 p_{10}= q^3 u^2 v c_1
 c_2 (q b_2 b_3 d_2 d_3+b_1b_2d_1 d_2+b_1b_3
 d_1d_3 ).
\end{gather*}
We give also explicit forms of the coefficients $c_{ij}$ $(0\leq j \leq 3, \, 0 \leq i+j \leq 4)$ of the polynomial $P(u,v)$ in variables $u$ and $v$:
\begin{gather*}
c_{01}= -c_0 \frac{ z^2}{q^2}
 \left(\frac{q}{b_1 d_1}+\frac{1}{b_2
 d_2}+\frac{1}{b_3 d_3}\right)
 \Phi(q z), \\
 c_{02}= c_0
 \frac{ z^2 (b_1 d_1+q
 b_2 d_2+q b_3 d_3) \Phi(q
 z)}{q^4 b_1 b_2 b_3 d_1 d_2
 d_3}, \qquad
 c_{03}= -c_0 \frac{
 z^2 \Phi(q z)}{q^5 b_1 b_2 b_3
 d_1 d_2 d_3}, \\
 c_{10}= c_0
 \left(\frac{z^2}{c_1}+\frac{z^2}{c_2}+\frac{z^2}{c_3} -z\right) \Phi(q z), \\
 c_{11}= -c_0z \bigg(\frac{(c_1+z) (c_2+z) (c_3+ z)\Phi(z)}{b_1 b_2 b_3 q} \\
\hphantom{c_{11}=}{}- \bigg( \frac{1}{b_1d_1q}+\frac{1}{b_3d_3q^2}+\frac{1}{b_2d_2q^2}+\frac{1}{b_1d_1q^2}+\frac{1}{b_3d_3q^3}+\frac{1}{b_2d_2q^3}\\
\hphantom{c_{11}=}{} +\frac{z^2}{b_1b_2b_3d_3q}
 +\frac{z^2}{b_1b_2b_3d_2q}
 +\frac{z^2}{b_1b_2b_3d_1q}\bigg) \Phi(qz) \\
 \hphantom{c_{11}=}{} + \bigg(
 \frac{d_1+d_2+d_3}{q^2 b_1b_2b_3d_1d_2d_3}z
 -\frac{b_1d_1+q b_2d_2+b_3d_3}{q^4 b_1b_2b_3 d_1d_2d_3}
 \bigg) \Phi\big(q^2z\big) -\frac{1}{b_1 b_2 b_3 d_1 d_2 d_3 q^4}\Phi \big(q^3 z\big) \bigg), \\
 c_{12}= c_0z \bigg( \bigg(
 \frac{d_1+d_2+d_3}{q^3 b_1b_2b_3 d_1d_2d_3}z
 -\frac{b_1d_1+qb_2d_2+qb_3d_3}{q^5 b_1b_2b_3d_1d_2d_3}
 \bigg) \Phi(qz) +\frac{(-1+q)}{q^5b_1b_2b_3d_1d_2d_3} \Phi\big(q^2z\big)\bigg),
 \\
 c_{13}= c_0z \frac{
 \Phi(q z)}{q^6 b_1 b_2 b_3 d_1
 d_2 d_3}, \\
 c_{20}= c_0z \bigg(
 \bigg(
 \frac{1}{c_1c_2}+\frac{1}{c_2c_3}+\frac{1}{c_3c_1}
 \bigg)z-
 \bigg( \frac{1}{c_1}+\frac{1}{c_2}+\frac{1}{c_3}
 \bigg)
 \bigg)
 \Phi(qz), \\
 c_{21}= c_0 \bigg(
 \frac{(c_1+z) (c_2+z) (c_3+ z)}{b_1 b_2 b_3 q^2} \Phi(z) \\
\hphantom{c_{21}=}{} - \bigg( \frac{(1+q)(d_1d_2+d_2d_3+d_3d_1)}{q^2b_1b_2b_3d_1d_2d_3}+
 \frac{1}{q^3b_3d_3}+\frac{1}{q^3b_2d_2}+\frac{1}{q^2b_1d_1} \bigg) \Phi(qz) \\
\hphantom{c_{21}=}{} +\bigg( \frac{d_1+d_2+d_3}{q^2b_1b_2b_3d_1d_2d_3 }z+
 \frac{b_1d_1 +q b_2d_2+qb_3d_3}{q^4 b_1b_2b_3d_1d_2d_3}\bigg)\Phi\big(q^2z\big)
 \bigg),\\
 c_{22}=
 -c_0 \frac{
 (d_1+d_2+d_3) \Phi(q
 z)}{q^4 b_1 b_2 b_3 d_1 d_2
 d_3}z\Phi(q
 z), \\
 c_{30}= c_0 \bigg( \frac{z^2}{c_1c_2c_3}-\bigg( \frac{1}{c_1c_2}+\frac{1}{c_2c_3}+\frac{1}{c_3c_1} \bigg) z
 \bigg) \Phi(qz),\\
 c_{31}= \frac{(d_1 d_2+d_2 d_3+d_3 d_1) z }{b_1 b_2 b_3
 d_1 d_2 d_3 q^2}\Phi(q z), \qquad
 c_{40}= -c_0 \frac{ z \Phi(q z)}{b_1
 b_2 b_3}.
\end{gather*}

\subsection*{Acknowledgements}

The author would like to express her gratitude to Professor Yasuhiko Yamada for valuable suggestions and encouragement.
And the author is grateful to referees for giving helpful comments to improve the manuscript.
She also thanks supports from JSPS KAKENHI Grant Numbers 17H06127 and 26287018 for the travel expenses in accomplishing this study.

\pdfbookmark[1]{References}{ref}
\LastPageEnding

\end{document}